\documentclass[aps,prl,twocolumn,10pt,notitlepage,showpacs]{revtex4-2}
\usepackage{amsmath}  
\usepackage{amsfonts}
\usepackage{amssymb}
\usepackage{amsthm}
\usepackage{braket}
\usepackage{algorithm2e}
\usepackage{fullpage}
\usepackage[]{nicefrac}
\usepackage{bbm}
\usepackage{qcircuit}
\usepackage{color}
\usepackage[normalem]{ulem}

\usepackage[final]{hyperref}
\hypersetup{
	colorlinks = true,
	allcolors = {blue},
}

\usepackage{cancel}
\usepackage{graphicx}
\usepackage{varwidth}

\DeclareMathOperator{\tr}{tr}

\def \eps {\varepsilon}
\newcommand{\diag}{\mathrm{diag}}
\newcommand{\cY}{{\cal Y}}

\def\Pr{\mathrm{Pr}}
\def\Exp{\mathbb{E}}
\usepackage{commath,mathtools}
\newcommand{\nrm}[1]{\left\lVert #1 \right\rVert}
\newcommand{\Oo}{\ensuremath{\mathcal{O}}}
\def\Tr{\mathrm{Tr}}

\newtheorem{corollary}{Corollary}
\newtheorem{lemma}{Lemma}

\begin{document}
\title{Quadratic speedup for spatial search by continuous-time quantum walk} 

\author{Simon Apers$^{1}$}
\email[]{smgapers@gmail.com}
\author{Shantanav Chakraborty$^{2}$}  
\email[]{shchakra@iiit.ac.in}
\author{Leonardo Novo$^{3}$}
\email[]{lfgnovo@gmail.com}
\author{J\'{e}r\'{e}mie Roland$^{3}$}
\email[]{jeremie.roland@ulb.ac.be}

\affiliation{$^1$ IRIF, CNRS, Paris, France}
\affiliation{$^2$ CQST and CSTAR, International Institute of Information Technology Hyderabad, Telangana, India} 
\affiliation{$^3$ QuIC, Ecole Polytechnique de Bruxelles, Universit\'{e} libre de Bruxelles, Brussels, Belgium}

\begin{abstract}
Continuous-time quantum walks provide a natural framework to tackle the fundamental problem of finding a node among a set of marked nodes in a graph, known as spatial search. Whether spatial search by continuous-time quantum walk provides a quadratic advantage over classical random walks has been an outstanding problem. Thus far, this advantage is obtained only for specific graphs or when a single node of the underlying graph is marked. In this article, we provide a new continuous-time quantum walk search algorithm that completely resolves this: our algorithm can find a marked node in any graph with any number of marked nodes, in a time that is quadratically faster than classical random walks. The overall algorithm is quite simple, requiring time evolution of the quantum walk Hamiltonian followed by a projective measurement. A key component of our algorithm is a purely analog procedure to perform operations on a state of the form $e^{-tH^2}\ket{\psi}$, for a given Hamiltonian $H$: it only requires evolving $H$ for time scaling as $\sqrt{t}$.
This allows us to quadratically fast-forward the dynamics of a continuous-time classical random walk. The applications of our work thus go beyond the realm of quantum walks and can lead to new analog quantum algorithms for preparing ground states of Hamiltonians or solving optimization problems. 
\end{abstract}
\date{\today}
\maketitle
Quantum walks, the quantum analogue of classical random walks, find widespread applications across quantum information processing. Besides being a universal model of quantum computation \cite{childs2009universal}, quantum walks have proven to be a useful primitive in the design of quantum algorithms \cite{ambainis2003quantum,childs2003exponential}. A particularly appealing model is that of continuous-time quantum walks (CTQWs) which provides a physics-inspired paradigm for quantum algorithm design, based solely on the evolution of a time-independent Hamiltonian that is related to a graph. One of the most widely studied algorithmic applications of this framework has been to tackle the spatial search problem of finding marked nodes on a graph. However, even though the first spatial search algorithm to tackle this problem was proposed almost two decades ago \cite{childs2004spatial}, it is unknown whether CTQWs can generally find a node within a marked set of nodes quadratically faster than classical random walks: a spatial search algorithm by CTQW that offers this general quadratic speedup has remained elusive. In contrast, in a recent breakthrough this problem was fully resolved for \emph{discrete-time} quantum walks \cite{ambainis2019quadratic}.\\
Childs and Goldstone provided the first spatial search algorithm by CTQW which could find a single marked node in certain graphs (such as the complete graph, hypercube and lattices of dimension greater than four) quadratically faster than classical random walks \cite{childs2004spatial}. Since then, a plethora of results have demonstrated a quadratic speed-up for several specific instances of graphs \cite{childs2004spatialdirac, childs2014spatial, janmark2014global, novo2015systematic, meyer2015connectivity, chakraborty2016spatial, chakraborty2017optimal, wong2018quantum, osada2020continuous, lewis2021optimal}. The approach by Childs and Goldstone is appealing because, similar to classical random walks, the state space of the walk is spanned by the nodes of the graph. However, explicit counterexamples are known where this approach fails to provide quadratic advantages for search, even for certain highly-connected graphs \cite{chakraborty2020optimality}. Furthermore, it seems difficult to analyse this algorithm in the scenario where multiple nodes in the underlying graph are marked and general results in this setting are missing.\\
More recently, Chakraborty et al.~developed a CTQW algorithm for spatial search on any weighted graph (or, equivalently, any ergodic, reversible Markov chain). This algorithm, which can be seen as a quantum walk on the edges of the underlying graph, offers a quadratic advantage over classical random walks only when a single node is marked \cite{chakraborty2020finding}. For certain instances with multiple marked nodes, the running time of the algorithm is even slower than its classical counterpart.\\
In this article we completely solve the spatial search problem by CTQW. We provide a CTQW-based algorithm that, starting from a quantum encoding of the stationary distribution $\pi$, finds a marked node on any graph, with any number of marked nodes in a time that is (up to a logarithmic factor) quadratically faster than the corresponding classical random walk. The key ideas behind this general result are as follows. We show that evolving a state $\ket{\psi_0}$ under a Hamiltonian $H$ for a random time proportional to $\sqrt{t}$ allows us, in a certain sense, to simulate projective measurements on the state $e^{-tH^2}\ket{\psi_0}$. When $H$ corresponds to the CTQW Hamiltonian on a graph, this procedure allows us to quadratically fast-forward the dynamics of a continuous-time random walk on $P$, analogous to the results of Ref.~\cite{apers2018FF}. The expected time required by a classical random walk to solve the spatial search problem is known as the hitting time $(HT)$. For the spatial search problem then it suffices to show that if we take $t \propto HT$, then the resulting state has a large overlap with the marked elements.\\
Obtaining a generic quadratic advantage for the spatial search problem has proven to be extremely challenging also in the discrete-time quantum walk (DTQW) framework and was only recently solved in \cite{ambainis2019quadratic}. However, this does not imply the existence of a CTQW algorithm with the same performance. In fact, even though there are general procedures to simulate Hamiltonian evolution (and hence CTQWs) using DTQWs \textit{\`{a} la} Szegedy \cite{childs2010relationship,berry2015hamiltonian}, there is no known way to translate DTQW dynamics to the continuous-time setting. The fact that these two frameworks are not equivalent is also why for the spatial search algorithm, efforts to prove a general quadratic speedup over classical random walks were undertaken in parallel.\\
Apart from being a fundamental contribution in the field of quantum walks, our work introduces new tools that are of independent interest and could lead to the development of novel analog quantum algorithms. In particular, our continuous-time procedure to access the quantum state $e^{-tH^2}\ket{\psi_0}$ naturally translates to a quantum algorithm to prepare the ground states of $H$. This procedure has a running time that matches the best quantum algorithms in the circuit model but is considerably simpler. We believe that similar techniques can lead to other analog quantum algorithms for optimization such as simulated annealing and Gibbs state preparation.

\textit{\textbf{Imaginary time-evolution --}} We begin by describing a technique of broader relevance, which uses a continuous-time procedure to prepare the state $e^{-tH^2}\ket{\psi}$, which can be seen as an imaginary time evolution under Hamiltonian $H^2$ for time $t$. In the spatial search setting, we will use this result to fast-forward the dynamics of a continuous-time random walk. Consider a quantum system in state $\ket{\psi_0}$ coupled to an ancillary system in a Gaussian state
\begin{equation}
\ket{\psi_g}=\int_{-\infty}^{+\infty} \dfrac{dz}{(2\pi)^{1/4}} e^{-z^2/4}\ket{z}.
\end{equation}
Choosing suitable units of mass and frequency, this state can be seen as the ground state of a one-dimensional quantum harmonic oscillator. The coupling is done via interaction Hamiltonian $H'=H\otimes\hat{z}$, where $\hat{z}$ corresponds to the position operator. Evolving $\ket{\psi_0}\ket{\psi_g}$, under $H'$ for a time $\sqrt{2t}$ results in the state
\begin{align}
\ket{\eta_t}&=e^{-i\sqrt{2t}H'}\ket{\psi_0}\ket{\psi_g} \nonumber \\
			\label{eqmain:evolved-state-fast-forwarding-1}
			 &=\int_{-\infty}^{+\infty}\dfrac{dz}{(2\pi)^{1/4}}e^{-z^2/4}e^{-i\sqrt{2t}Hz}\ket{\psi_0}\ket{z} \\		 
			 &=\int_{-\infty}^{+\infty}\dfrac{dz}{\sqrt{2\pi}}e^{-z^2/2}e^{-i\sqrt{2t}Hz}\ket{\psi_0}\ket{\psi_g}+\ket{\Phi}^\perp, \nonumber
\end{align}
where $\ket{\Phi}^\perp$ is a quantum state where the ancillary system is orthogonal to $\ket{\psi_g}$. Now we can make use of the Hubbard-Stratonovich transformation which states that, for any $x \in \mathbb{R}$,
\[
e^{-x^2/2}=\int \dfrac{dy}{\sqrt{2\pi}}~e^{-y^2/2}e^{-ixy}.
\]  
If we choose $x=\sqrt{2t}H$, it allows us to rewrite
\begin{equation}\label{eqmain:evolved-state-fast-forwarding-3}
    \ket{\eta_t}=e^{-tH^2}\ket{\psi_0}\ket{\psi_g}+\ket{\Phi}^\perp.
\end{equation}
By post-selecting on obtaining $\ket{\psi_g}$ in the second register (for example, by measuring the second register in the eigenbasis of the quantum harmonic oscillator), we are able to prepare a quantum state proportional to $e^{-tH^2}\ket{\psi_0}$ in the first register.\\
This simple analog procedure in itself can have several applications. For example, consider the problem of preparing the ground state of $H$ (say $\ket{v_0}$). If $\Delta$ is the gap between the ground state and the first excited state of $H$, then it can be shown that the state $e^{-tH^2}\ket{\psi_0}$, for $\sqrt{t}$ scaling as $\Oo\left(\frac{1}{\Delta}\log^{1/2}\left(\frac{1}{ |\braket{\psi_0|v_0}|\eps}\right)\right)$ \footnote{We use the Big-Oh notation, i.e.\ $f(n)\in \mathcal{O}(g(n))$ if there exists a constant $c$ such that $f(n)\leq c g(n)$. Any $f(n)\in\Omega(g(n))$ if there exists a constant $c$ such that $f(n)\geq c g(n)$. Finally $f(n)\in \Theta(g(n))$, if there exists constants $c_1,c_2$ such that $c_1~g(n)\leq f(n)\leq c_2~g(n)$.}, is $\eps$-close to $\ket{v_0}$ with probability scaling as $|\braket{\psi_0|v_0}|^2$ (See Appendix for details). The performance matches the best circuit model algorithms for this problem but is conceptually simpler as unlike in the discrete-time setting, does not require implementing linear combination of unitaries or quantum phase estimation \cite{abrams1999quantum,ge2019faster, lin2020near}.\\
We believe that our approach can also lead to potentially simpler analog quantum algorithms for preparing Gibbs states, as in \cite{chowdhury2017quantum}.\\

\textit{\textbf{Dropping the ancilla register --}} For several quantum algorithmic applications, we are only interested in the projection of the state $e^{-tH^2}\ket{\psi_0}$ into some subspace of interest. This is typically the case for simulated annealing, where we are only interested in finding a (close to) optimal solution. Similarly it holds for the spatial search problem, wherein $H$ corresponds to the quantum walk Hamiltonian, and we want the final state to have a good overlap with the marked nodes. In such scenarios, we will show that it is possible to get rid of the ancilla register altogether: it suffices to evolve according to the Hamiltonian $H$ for some appropriately random time.
We formalize this in the following lemma.
\begin{lemma}
\label{lem:randomized-time-evolution}
\textit{Evolving a quantum state $\ket{\psi_0}$ under a Hamiltonian $H$ for time $\sqrt{2t}z$, where $z\sim \mathcal{N}(0,1)$ is drawn from the standard normal distribution, results in a mixed state
\[
\rho_t=\int_{-\infty}^{+\infty} \dfrac{dz}{\sqrt{2\pi}}e^{-z^2/2}e^{-i\sqrt{2t}Hz}\ket{\psi_0}\bra{\psi_0}e^{i\sqrt{2t}Hz},  
\]
such that for any projector} $\Pi$, 
$$\tr[\Pi\rho_t]\geq \bra{\psi_0}e^{-tH^2}\Pi e^{-tH^2}\ket{\psi_0}.$$
\end{lemma}
\begin{proof}
It can be seen that 
\begin{align}
\tr[\Pi\rho_t]&= \tr\left[\left(\Pi\otimes I\right)\ket{\eta_t}\bra{\eta_t}\right] \nonumber \\
&=\bra{\psi_0}e^{-tH^2}\Pi e^{-tH^2}\ket{\psi_0}+ \bra{\Phi^\perp}\Pi\otimes I\ket{\Phi^\perp}, \nonumber
\end{align}
where the first equality follows from Eq.~\eqref{eqmain:evolved-state-fast-forwarding-1} and the second from Eq.~\eqref{eqmain:evolved-state-fast-forwarding-3}. The proof of the Lemma then follows from the fact that the last term in the expression for $\tr[\Pi\rho_t]$ is non-negative.
\end{proof}
The expected time required to obtain $\rho_t$ is simply $T=\sqrt{2t}\cdot \langle |z|\rangle$.
Since $z\sim \mathcal{N}(0,1)$ this is $T=2\sqrt{t/\pi}$.
We note that this lemma is an analog variant of an observation in \cite[Section 6]{apers2019unified}.
Next we explain how Lemma \ref{lem:randomized-time-evolution} can be used to fast-forward continuous-time random walk dynamics - a key ingredient of our spatial search algorithm.\\

\textit{\textbf{Fast-forwarding continuous-time random walks} --}
Before explaining the construction of the Hamiltonian driving the quantum walk, we introduce some notions about Markov chains and random walks. For details on some of these concepts, we refer the reader to the Appendix. A Markov chain with a set of $V$ nodes, such that $|V|=n$, is defined by an $n\times n$ stochastic matrix $P$, such that $P_{xy} = p_{xy}$, which is the probability of of transitioning from node $x$ to node $y$. The state of the random walker at any time is given by a stochastic row vector. A discrete-time random walk proceeds by repeated applications of $P$, starting from some initial state $v_0$: $t$ random walk steps result in the state $v_t=v_0P^t$.\\
To any discrete-time chain with transition matrix $P$, we can associate a continuous-time chain by choosing the \emph{transition kernel} $Q=P-I$. Starting from some distribution $v_0$, the state of this continuous-time walk after time $t$ is described by the distribution $v_t=v_0e^{Qt}$.\\ 
We assume that $P$ represents an ergodic and reversible Markov chain, implying (i) the eigenvalues of $P$ lie between $-1$ and $1$ and (ii) $P$ has a unique stationary distribution $\pi=(\pi_1~\pi_2~\cdots~\pi_n)$, which obeys the so-called detailed balance condition, $\pi_x p_{xy}=\pi_y p_{yx}$. Any reversible $P$ corresponds to a discrete-time random walk on a weighted graph. A useful matrix associated to $P$ is the discriminant matrix $D$, defined as an $n\times n$ symmetric matrix with entries $D_{xy}=\sqrt{p_{xy}p_{yx}}$. The discriminant matrix is symmetric, its spectrum coincides with that of $P$, and the coherent encoding of the stationary distribution $\ket{\sqrt{\pi}}=\sum_{j}\sqrt{\pi_j}\ket{j}$ is an eigenstate of $D$ with eigenvalue $1$.\\
The discriminant matrix provides a framework to deal with cases when $P$ is not symmetric and the spectral relationship between the two has been fundamental for designing many quantum walk algorithms. In fact, for symmetric $P$, $D=P$.\\
In particular, it is used in a canonical construction of a Hamiltonian $H_P$, associated to a reversible Markov chain $P$, so that $e^{-iH_Pt}$ corresponds to a quantum walk on the edges of the underlying graph for time $t$~\cite{somma2010quantum, chakraborty2020finding}.
A useful property is that the spectra of $H_P$ and $P$ are closely related.
To define $H_P$, consider the Hilbert space $\mathcal{H}$ spanned by the nodes of $P$ and an additional reference state $\ket{0}$.
We define (i) a unitary $U_P$ acting on $\mathcal{H}\otimes \mathcal{H}$ such that $U_P\ket{x,0}=\sum_{y\in V}\sqrt{p_{xy}}\ket{x,y}$, for any $x\in V$, and (ii) a swap operation $S$ defined by $S \ket{x,y} = \ket{y,x}$.
The CTQW Hamiltonian is then defined as the commutator $H_P=i[U_P^\dag S U_P,\Pi_0]$, where $\Pi_0=I\otimes\ket{0}\bra{0}$.\\
We refer the readers to Refs.~\cite{krovi2010adiabatic,chakraborty2020finding} wherein the spectral properties of $H_P$ have been analyzed in detail. Here, we shall rely only on an elegant relationship between $H_P$ and the discriminant matrix of the underlying Markov chain. Specifically, for any quantum state $\ket{\psi}\in \mathcal{H}$ it holds that
\begin{equation}
H^2_P\ket{\psi,0}=I-D^2\ket{\psi,0}. \nonumber
\end{equation}
We now combine this with Lemma \ref{lem:randomized-time-evolution}.
It follows that evolving any state $\ket{\psi_0}=\ket{\psi,0}$ under $H_P$ for a time $T=2\sqrt{t}z$, where $z\sim \mathcal{N}(0,1)$, results in a mixed state $\rho_t$ such that
\begin{equation}
\tr[\Pi\rho_t]\geq \bra{\psi_0}e^{t\left(D^2-I\right)}\Pi e^{t\left(D^2-I\right)}\ket{\psi_0}. \nonumber
\end{equation}
Note that $e^{D^2-I}$ has the same spectrum as $e^{P^2-I}$ which is the continuous-time random walk transition matrix obtained from $P^2$, corresponding to two steps of the discrete-time random walk. The previous inequality can be interpreted as follows: if we want to find a marked node by measuring $e^{t\left(D^2-I\right)}\ket{\psi_0}$, then we may instead measure $\rho_t$, for which the probability of finding a marked node is at least as large. As the expected time to prepare $\rho_t$ is $\langle |T| \rangle = 2\sqrt{t/\pi}$, this can be seen as a way to fast-forward a continuous-time random walk, in analogy to the discrete-time quantum fast-forwarding scheme introduced in Ref~\cite{apers2019quantum}. This technique is a crucial ingredient for the proof of a quadratic speedup for spatial search. We will show that when $t$ scales as the hitting time of the classical walk, and we make a suitable choice of $D$, then a marked node is found with large probability by measuring $\rho_t$.\\

\textit{\textbf{Spatial search by CTQW}--} In the classical scenario, the problem of searching marked nodes by random walk can be stated as follows. Given a reversible Markov chain $P$ with $|V|=n$ nodes, suppose $M\subset V$ corresponds to the set of marked nodes. Starting from its stationary distribution $\pi$, a random walk search algorithm repeatedly applies a step of the walk, while checking whether the new node is marked. The expected time required by the random walk to find some node $x\in M$ is defined as the hitting time $HT(P,M)$. The classical algorithm stops as soon as a marked node is reached.
As such, it be seen as the repeated application of the absorbing walk $P'$, obtained from $P$ by replacing all the outgoing edges from nodes in $M$ with self-loops.\\
The original approach by Szegedy to obtain speed-ups for spatial search problems used a quantum walk operator based on from $P'$ \cite{szegedy2004quantum}.
However, improvements over this approach made use of the framework of \emph{interpolated Markov chains} \cite{krovi2010adiabatic, krovi2016quantum,ambainis2019quadratic,chakraborty2020optimality}, and we will also use this idea in our work.
For an ergodic reversible Markov chain $P$ and absorbing walk $P'$, the interpolated random walk operator is given by $P(s)=(1-s)P+sP'$, where $s\in [0,1]$. Intuitively, at any timestep the interpolated random walk first checks whether the current node is marked and (i) if it is not, the original walk operator $P$ is applied, while (ii) if it is, then $P$ is applied with probability $(1-s)$ and otherwise $P'$ is applied (i.e., the walk remains at the node). The advantage with using this framework is that for appropriate choice of $s$, the interpolated walk $P(s)$ demonstrates the useful properties of both the original Markov chain $P$ as well as the absorbing chain $P'$.
E.g.,  Krovi et al.~\cite{krovi2016quantum} showed that if $P$ is ergodic and reversible, then so is $P(s)$ for any $s\in[0,1)$. So $P(s)$ has a unique stationary distribution $\pi(s)$ and its spectrum coincides with its discriminant matrix $D(s)$.\\
We now have all the necessary ingredients to introduce our quantum search algorithm. Similar to prior works, we assume we have a procedure to prepare the state $\ket{\sqrt{\pi}}$, as well as a procedure to measure if a node is marked or not. We define $\Pi_M$ as the projective measurement onto the set of marked nodes $M$. The first step of the algorithm is to prepare $\ket{\sqrt{\pi}}$ and perform a projective measurement associated to operators $\{\Pi_M, I-\Pi_M\}$.
Either this returns a marked node, in which case we are done, or it returns the post-measurement state $\ket{\sqrt{\pi_U}}=\sum_{j \in U} \sqrt{\pi_j}\ket{j}$, where $U = V \setminus M$ denotes the set of nodes that are not marked.\\
The next step is to evolve $\ket{\sqrt{\pi_U}}$ according to the Hamiltonian $H_{P(s)}$, for some $s\in [0,1)$. More specifically, we use a random time evolution so as to prepare the state $\rho_t$ as in Lemma~\ref{lem:randomized-time-evolution}. Using this Lemma, we obtain that the probability of observing a marked node by measuring $\rho_t$ is $\Tr[(\Pi_M\otimes I)\rho_t]\geq \nrm{\Pi_M e^{(D^2(s)-I)t}\ket{\sqrt{\pi_U}}}^2$. Now recall that $\rho_t$ can be prepared in $O(\sqrt{t})$ time.
Hence, this results in a quadratic speedup over classical random walks if we can choose some $t$ scaling as $HT(P,M)$, and some $s\in [0,1)$, so that the success probability is large. While the optimal choice of $s$ seems hard to determine, we show that a large success probabilities can be obtained by choosing an appropriately random $s$, similar to Ref.~\cite{ambainis2019quadratic}. This essentially follows from the following lemma.
\begin{lemma}
\label{lem:success-probability-lower-bound}
\textit{Consider an ergodic, reversible Markov chain $P$ and a set of marked nodes $M$. If we choose parameters $s\in \{1-1/r: r=1,2,\cdots, 2^{\lceil\log T\rceil}\}$ and $T\in \Theta\left(HT(P,M)\right)$ uniformly at random, then the continuous-time random walk operator $e^{D^2(s)-I}$ satisfies}
\begin{equation}
\mathbb{E}\left[\nrm{\Pi_Me^{(D^2(s)-I)T}\ket{\sqrt{\pi_U}}}^2\right]\in \Omega\left(1/\log^2 T\right). \nonumber
\end{equation} 
\end{lemma}
Here, we summarize the main technical steps of the proof, while the details are provided in the Appendix. First, we show that the probability of finding a marked state by measuring the (sub-normalized) quantum state  
$e^{(D^2(s)-I)T}\ket{\sqrt{\pi_U}}$ can be lower bounded by a purely classical quantity. This quantity is the probability of a specific event happening in a continuous-time random walk with walk operator $e^{(P^2(s)-I)T}$. Then, through some non-trivial reductions, we show that this probability on its turn can be lower bounded by the probability of the same event happening in a discrete-time walk with walk operator $P(s)$. This reduction to a problem involving discrete-time walks allows us to use the central result of Ref.~\cite{ambainis2019quadratic} to finally prove Lemma~\ref{lem:success-probability-lower-bound}.\\
The final CTQW search algorithm (Algorithm \ref{algo:spatial-search-algorithm}) is described below.\\
\RestyleAlgo{boxruled}
\begin{algorithm}[ht]
  \caption{Quantum spatial search by CTQW}\label{algo:spatial-search-algorithm}
  ~Pick some $s\in \{1-1/r: r=1,2,\cdots, 2^{\lceil\log T\rceil}\}$ and $T\in \Theta\left(HT(P,M)\right)$ uniformly at random.
  \begin{itemize}
  \item[1.~] Prepare the state $\ket{\sqrt{\pi}}\ket{0}$.\\
  \item[2.~] Perform a measurement using $\{\Pi_M,I-\Pi_M\}$ on the first register. If the output is marked, measure in the node basis to obtain $x\in M$ and output $x$. Otherwise we are left with the state $\ket{\sqrt{\pi_U}}\ket{0}$. \\
  \item[3.~] Evolve $\ket{\sqrt{\pi_U}}\ket{0}$, under $H_{P(s)}$ for time $\sqrt{2T}z$, where $z\sim\mathcal{N}(0,1)$, to prepare the state $\rho_t$.\\
  \item[4.~] Perform a measurement on $\rho_t$ in the node basis on the first register. 
  \end{itemize}
\end{algorithm}
\\
By combining Lemma~\ref{lem:randomized-time-evolution} and Lemma~\ref{lem:success-probability-lower-bound}, we see that Algorithm \ref{algo:spatial-search-algorithm} runs in expected time $\Theta(\sqrt{HT(P,M)})$, and has a success probability of at least $1/\log^2 \left(HT(P,M)\right)$. Thus by repeating it $\log^2 \left(HT(P,M)\right)$ times, we obtain a marked node $x\in M$ with a constant probability.
This yields an overall expected runtime of
$$T=\Theta\left(\sqrt{HT(P,M)}\log^2\left(HT(P,M)\right)\right).$$

\textbf{\textit{Discussion --}} We have presented a CTQW search algorithm that achieves a general quadratic speedup over classical random walks on an arbitrary graph, even when multiple nodes are marked. This solves a long-standing open problem in the field. The running time of our algorithm matches that of a recent breakthrough DTQW algorithm for this problem~\cite{ambainis2019quadratic}.\\
This equivalence in performance between DTQWs and CTQWs for the spatial search problem raises the question of whether more general equivalences can be proven between these different frameworks. In particular, it is currently unknown if any DTQW algorithm or, more generally, any dynamics generated by DTQW can be simulated by a CTQW (although the reverse holds owing to Childs' reduction \cite{childs2010relationship}). It would be interesting to explore whether such an equivalence can be obtained: solving this may bring new insights into the power of simple analog models of computation and inspire new analog quantum algorithms.\\
Indeed, the development of our CTQW algorithm for spatial search has led to the development of tools of independent interest. For any Hamiltonian $H$, we provide a simple continuous-time procedure to prepare the state $e^{-tH^2}\ket{\psi_0}$. This translates to an analog quantum algorithm for preparing the ground states of $H$ with a running time that matches the best known quantum algorithms in the circuit model. For several physical systems, preparing ground states of Hamiltonians is of extreme interest as they contain information about interesting phases of matter and are also useful to solve optimization problems \cite{lucas2014ising}. In this regard, our work could also lead to other simple analog quantum algorithms for optimization such as quantum simulated annealing \cite{somma2008quantum} and preparation of thermal states \cite{chowdhury2017quantum}.\\~\\
\textit{Note: After the completion of this article, we became aware of recent Refs.~\cite{keen2021quantum} and~\cite{he2021quantum}. In~\cite{keen2021quantum}, the authors have applied a discretized version of the Hubbard-Stratonovich transform (using linear combination of unitaries) to prepare the ground state of a Hamiltonian in the circuit model. Ref.~\cite{he2021quantum} explores ideas similar to ours to develop analog near-term quantum algorithms for ground state preparation. Neither of these references consider the spatial search problem.}
\begin{acknowledgments}
SA acknowledges funding from QuantERA ERA-NET Cofund project QuantAlgo. SC acknowledges support from the Faculty Seed Grant, IIIT Hyderabad. LN acknowledges support from F.R.S.- FNRS. JR is supported by the Belgian Fonds de la Recherche Scientifique - FNRS under grant no  R.50.05.18.F (QuantAlgo). 
\end{acknowledgments}
 
\bibliography{bibliography}
\bibliographystyle{unsrt}
\widetext
\clearpage
\begin{center}
\textbf{\large Appendix for Quadratic speedup for spatial search by continuous-time quantum walk}
\end{center}
\setcounter{equation}{0}
\setcounter{figure}{0}
\setcounter{table}{0}
\setcounter{algocf}{0}
\makeatletter
\renewcommand{\theequation}{S\arabic{equation}}
\renewcommand{\thefigure}{S\arabic{figure}}
\renewcommand{\thelemma}{S\arabic{lemma}}
\renewcommand{\thealgocf}{S\arabic{algocf}}
\renewcommand{\thecorollary}{S\arabic{corollary}}
Here we elaborate on the proofs of the main article.
\section{Discrete and continuous-time Markov chains}
\label{subsec:basics-mc}
\textbf{Discrete-time Markov chains}:~ A sequence of random variables $(Y)=(Y_j)_{j=0}^{\infty}$, is a (discrete-time) Markov chain if for all $j>0$,
\begin{equation}
\Pr\left[Y_i=y_i|Y_0=y_0,Y_1=y_1,\cdots,Y_{i-1}=y_{i-1}\right]=\Pr\left[Y_i=y_i|Y_{i-1}=y_{i-1}\right]. \nonumber
\end{equation} 
A discrete-time Markov chain over a set of $V$ nodes, such that $|V|=n$, can be described by a $n\times n$ stochastic matrix $P$. Each entry $p_{xy}$ of this matrix $P$, known as the transition matrix, represents the probability of transitioning from state $x$ to state $y$.\\
A \textit{discrete-time random walk} on such a chain is defined as follows: Suppose $v_0$ is some distribution over $V$ and $Y_0$ is the random variable distributed according to $v_0$. Then $t$ steps of a discrete-time random walk results in a distribution $v_t=v_0P^t$ such that $Y_t$ is a random variable distributed according to $v_t$. Equivalently, $t$-steps of a discrete-time random walk corresponds to a sequence of random variables $(Y)=(Y_1,Y_2,\cdots,Y_t)$, such that each $Y_i$ is a random variable from the distribution $v_i$.\\
The probability that, starting from the distribution $v$, a discrete-time Markov chain $(Y)$ belongs to some subset $M$ of $V$ after $t$-steps will be denoted as $\Pr_{v}\left[Y_t\in M\right]$.\\~\\
\textbf{Properties of Markov chains:~}A Markov chain is \textit{irreducible} if any state can be reached from any other state in a finite number of steps. Any \textit{irreducible} Markov chain is \textit{aperiodic} if there exists no integer greater than one that divides the length of every directed cycle of the graph. A Markov chain is \textit{ergodic} if it is both \textit{irreducible} and \textit{aperiodic}. By the Perron-Frobenius Theorem, any ergodic Markov chain $P$ has a unique stationary state $\pi$ such that $\pi P=\pi$. The stationary state $\pi$ is a stochastic row vector and has support on all the nodes in $V$. Let us denote it as 
\begin{equation}
\pi=\left(\pi_1~~\pi_2~~\cdots~~\pi_n\right), \nonumber
\end{equation}
such that $\sum_{j=1}^n \pi_j=1$. For any ergodic Markov chain, $\pi$ is the unique eigenvector with eigenvalue $1$. All eigenvalues of $P$ lie between $-1$ and $1$. We shall concern ourselves with ergodic Markov chains that are also \textit{reversible}, i.e. Markov chains which satisfy the \textit{detailed balance condition} 
$$\pi_x p_{xy}=\pi_y p_{yx},~\forall (x,y)\in X.$$ 
This can also be rewritten as 
$$\text{diag}(\pi)P=P^T\text{diag}(\pi),$$
where $\text{diag}(\pi)$ is a diagonal matrix with the $j^{\text{th}}$ diagonal entry being $\pi_j$. In other words, the reversibility criterion implies that the matrix $\text{diag}(\pi)P$ is symmetric. Henceforth we shall only deal with reversible (and hence ergodic) Markov chains.
~\\~\\
\textbf{Discriminant matrix:~} The discriminant matrix of $P$ is a symmetric matrix $D$ such that the $(x,y)^{\mathrm{th}}$ entry of $D$ is $D_{xy}=\sqrt{p_{xy}p_{yx}}$. For any ergodic, reversible Markov chain $P$, 
\begin{equation}
D=\mathrm{diag}(\sqrt{\pi})P\mathrm{diag}(\sqrt{\pi})^{-1}, \nonumber
\end{equation}
where $\sqrt{\pi}$ is a row vector with its $j^{\text{th}}$-entry being $\sqrt{\pi_j}$. From this it is easy to see that $D$ and $P$ have the same set of eigenvalues. In particular, the eigenstate of $D$ with eigenvalue $1$ is given by
\begin{equation}
\ket{\pi}=\sum_{j=1}^n\sqrt{\pi_j}\ket{j}. \nonumber
\end{equation}
\textbf{Lazy random walk:} Often, it is convenient to add self-loops to every node of $P$, resulting in what is called the \textit{lazy walk}. In fact, we shall map $P\mapsto (I+P)/2$ such that $P_{ii}=1/2$ for $1\leq i\leq n$. This transformation ensures that for the Markov chain corresponding to the lazy walk, all its eigenvalues lie between $0$ and $1$. It will be easier to deal with lazy walks, which do not affect our results. For example, the hitting time is affected by a factor of two by this transformation.\\~\\
\textbf{Interpolated Markov chains:} Let $M\subset V$ denote the set of marked nodes of $P$. For any $P$, we define $P'$ as the \textit{absorbing Markov chain} obtained from $P$ by replacing all the outgoing edges from $M$ by self-loops. If we re-arrange the elements of $V$ such that the set of unmarked nodes $U:=V\backslash M$ appear first, then we can write 
\begin{align}
P=\begin{bmatrix}
P_{UU} & P_ {UM}\\
P_{MU} & P_{MM}
\end{bmatrix},~~~~~~~~~P'=\begin{bmatrix}
P_{UU} & P_ {UM}\\
0 & I
\end{bmatrix}, \nonumber
\end{align}
where $P_{UU}$ and $P_{MM}$ are square matrices of size $(n-|M|)\times (n-|M|)$ and $|M|\times |M|$ respectively. On the other hand $P_{UM}$ and $P_{MU}$ are matrices of size $(n-|M|)\times |M|$ and $|M|\times (n-|M|)$ respectively. 

An \textit{interpolated Markov chain}, $\left(X^{(s)}\right)$ corresponds to a transition matrix
\begin{equation}
P(s)=(1-s)P+sP', \nonumber
\end{equation}
where $s\in[0,1]$. $P(s)$ thus has a block structure
\begin{align}
P=\begin{bmatrix}
P_{UU} & P_ {UM}\\
(1-s)P_{MU} & (1-s)P_{MM}+sI
\end{bmatrix}. \nonumber
\end{align}
Clearly, $P(0)=P$ and $P(1)=P'$. Notice that if $P$ is ergodic, so is $P(s)$ for $s\in [0,1)$. This is because any edge in $P$ is also an edge of $P(s)$ and so the properties of \textit{irreducibility} and \textit{aperiodicity} are preserved. However when $s=1$, $P(s)$ has outgoing edges from $M$ replaced by self-loops and as such the states in $U$ are not accessible from $M$, implying that $P(1)$ is not ergodic. 

Now we shall see how the stationary state of $P$ is related to that of $P(s)$. Since $V=U \cup M$, the stationary state $\pi$ can be written as $\pi=(\pi_U~~\pi_M)$, where $\pi_U$ and $\pi_M$ are row-vectors of length $n-|M|$ and $|M|$ obtained by projecting the stationary distribution on to the marked subspace and the unmarked subspace, respectively. As mentioned previously, $P'$ is not ergodic and does not have a unique stationary state. In fact, any state having support over only the marked set is a stationary state of $P'$. On the other hand $P(s)$ is ergodic for $s\in [0,1)$. Let $p_M=\sum_{x\in M}\pi_x$ be the probability of obtaining a marked node in the stationary state of $P$. Then it is easy to verify that the unique stationary state of $P(s)$ is 
\begin{equation}
\pi(s)=\dfrac{1}{1-s(1-p_M)}\left((1-s)\pi_U~~\pi_M\right). \nonumber
\end{equation}
\textbf{Continuous-time Markov chains:} A continuous-time Markov chain is defined by a continuous sequence of random variables $(X)$. From a discrete-time chain with transition matrix $P$ we can derive a continuous-time chain by setting $Q = P - I$. If $P$ has no self-loops then the continuous-time chain simply follows by replacing every state of the discrete-time chain by an interval of length distributed according to an exponential random variable with rate $1$. If $P$ is lazy ($P_{ii} = 1/2$ for every $i$) then the rate of the exponential random variables is $1/2$.
This implies that
\begin{equation}
\Pr(X_t = x \mid X_0 = y)
= (e^{Q t})_{xy}, \nonumber
\end{equation}
which corresponds to the $(x,y)^{\text{th}}$-entry of the matrix $e^{tQ}$. We refer to $Q$ as the transition matrix of the continuous-time random walk on Markov chain $(X)$. So, a continuous-time random walk for some time $t$, starting from some distribution $v_0$ is simply $v_t=v_0 e^{Qt}$. 

Note that the random walk maps a distribution over $V$ to another distribution over $V$ as $Q$ is a matrix whose rows sum to $0$. Equivalently a continuous-time random walk on a graph for time $t$, corresponds to a continuous sequence of random variables over $V$, i.e. $X=(X_t)$. One can also define interpolated continuous-time Markov chains (or random walks) by considering $P(s)$: definitions translate naturally. Like in the discrete-time scenario, the probability that, starting from the distribution $v$ a continuous-time Markov chain $(X)$ belongs to some subset $M$ of $V$ after time $t$ will be denoted as $\Pr_{v}\left[X_t\in M\right]$.

Finally, it will be useful for us to consider the continuous-time random walk with transition matrix $Q=P^2-I$. This corresponds to the walk obtained from two steps of a discrete-time random walk.
\section{Proof of lemma 2}

We want to prove the following lemma (also stated in the main article):\\~\\

\textbf{Lemma 2.~}
\textit{Consider an ergodic, reversible Markov chain $P$ and a set of marked nodes $M$. If we choose parameters $s\in \{1-1/r: r=1,2,\cdots, 2^{\lceil\log T\rceil}\}$ and $T\in \Theta\left(HT(P,M)\right)$ uniformly at random, then the continuous-time random walk operator $e^{D^2(s)-I}$ satisfies}
\begin{equation}
\mathbb{E}\left[\nrm{\Pi_Me^{(D^2(s)-I)T}\ket{\sqrt{\pi_U}}}^2\right]\in \Omega\left(1/\log^2 T\right). \nonumber
\end{equation} 

We begin by relating the quantity $\nrm{\Pi_Me^{(D(s)^2-I)T}\ket{\pi_U}}^2$ to the behaviour of the original Markov chain $P(s)$ via the following lemma (which applies to any reversible Markov chain).

\begin{lemma} \label{lemsup:success-prob}
\textit{Consider a continuous-time Markov chain $(X)$ with transition matrix $P^2 - I$, where $P$ is a reversible Markov chain transition matrix with stationary distribution $\pi$ and discriminant matrix $D$.
Let $M$ be a set of marked nodes and $U$ be the set of non-marked nodes of $P$.
Then for any $t,t' \geq 0$ we have that
\[
\| \Pi_M e^{t(D^2-I)} \ket{\sqrt{\pi_U}} \|
\geq \Pr_{\pi_U}(X_t \in M, X_{t+t'} \notin M).
\]}
\end{lemma}
\begin{proof}
The proof is similar to \cite[Lemma 8]{ambainis2019quadratic}. Since $\| \Pi_M e^{t' (D^2-I)} \| \leq 1$ for any $t' \geq 0$, we can bound
\begin{align*}
\| \Pi_M e^{t(D^2-I)} \ket{\sqrt{\pi_U}} \|
&\geq \| \Pi_M e^{t(D^2-I)} \ket{\sqrt{\pi_U}} \| \, \| \Pi_M e^{t' (D^2-I)} \ket{\sqrt{\pi_U}} \| \\
&\geq \bra{\sqrt{\pi_U}} e^{t(D^2-I)} \Pi_M e^{t' (D^2-I)} \ket{\sqrt{\pi_U}},
\end{align*}
where the second inequality follows from Cauchy-Schwarz.
Using that $D = \diag(\sqrt{\pi}) P \diag(\sqrt{\pi})^{-1}$ we can rewrite this as
\begin{align}
\bra{\sqrt{\pi_U}} e^{t(D^2-I)} &\Pi_M e^{t' (D^2-I)} \ket{\sqrt{\pi_U}} \nonumber \\
&= \bra{\sqrt{\pi_U}} \diag(\sqrt{\pi}) e^{t(P^2-I)} \diag(\sqrt{\pi})^{-1} \Pi_M \diag(\sqrt{\pi}) e^{t' (D^2-I)} \diag(\sqrt{\pi})^{-1} \ket{\sqrt{\pi_U}} \nonumber \\
&= \bra{\pi_U} e^{t(P^2-I)} \Pi_M e^{t'(P^2-I)} \ket{1_U}
= \Pr_{\pi_U}(X_t \in M, X_{t+t'} \notin M), \nonumber
\end{align}
where we used the shorthands $\ket{\pi_U} = \frac{1}{\pi(U)} \sum_{x \in U} \pi(x) \ket{x}$ and $\ket{1_U} = \sum_{y \in U} \ket{y}$.
\end{proof}
Note that as $P(s)=(1-s)P+sP'$ is reversible for any $s\in[0,1)$, we can extend Lemma \ref{lemsup:success-prob} to any interpolated Markov chain. We shall now be working with this Markov chain quantity for interpolated Markov chains $(X^{(s)})$ with transition matrix $P^2(s) - I = ((1-s) P + s P')^2 - I$.\\
Next, we use a probabilistic argument to show that there exist some (in fact, many) choices of $t$, $t'$ and interpolation parameter $s$ for which the Markov chain quantity can be lower bounded. Specifically, fix some finite subset $R \subset [0,1)$ and a range $[0,T]$. If we pick an interpolation parameter $s \in R$ and times $t,t' \in [0,T]$ uniformly at random, then we can lower bound the success probability by
\begin{align}
\mathbb{E}_{s,t}\left[ \| \Pi_M e^{t((D{(s)})^2-I)} \ket{\sqrt{\pi_U}} \|^2 \right]
&\geq \mathbb{E}_{s,t,t'}\left[ \Pr_{\pi_U}(X^{(s)}_t \in M, X^{(s)}_{t+t'} \notin M)^2 \right] \nonumber \\
\label{eqsup:lower-bound-succ-prob}
&\geq \mathbb{E}_{s,t,t'}\left[ \Pr_{\pi_U}(X^{(s)}_t \in M, X^{(s)}_{t+t'} \notin M) \right]^2,
\end{align}
where we used Cauchy-Schwarz in the last inequality.
We will lower bound this by using the central technical lemma from Ref.~\cite{ambainis2019quadratic}, which is about \emph{discrete-time} random walks.
\begin{lemma}[Discrete-time, {\cite[Corollaries 11-12]{ambainis2019quadratic}}] \label{lem:amb}
\textit{Let $P$ be a reversible ergodic discrete-time Markov chain with stationary distribution $\pi$.
Let $(Y^{(s)})$ denote the discrete-time Markov chain with transition matrix $P(s)$.
If $\pi(M) \leq 1/9$ and $T \geq 3 \mathrm{HT}(M)$, and we choose interpolation parameter $s \in \{1 - 1/r \,:\, r = 1,2,\dots,2^{\lceil \log(12T)\rceil}\}$ and time $t,t' \in [24T]$ uniformly at random, then}
\[
\mathbb{E}_{s,t,t'}\left[ \Pr_{\pi_U}(Y^{(s)}_t \in M, Y^{(s)}_{t+t'} \notin M) \right]
\in \Omega(1/\log(T)).
\]
\end{lemma}
Given Eq.~\eqref{eqsup:lower-bound-succ-prob} and Lemma \ref{lem:amb}, we can prove Lemma 2 by showing that
\begin{equation}
\label{eqsup:required-lower-bound-expected-prob}
\mathbb{E}_{s,t,t'}\left[\Pr_{\pi_U}(X^{(s)}_t \in M, X^{(s)}_{t+t'} \notin M)\right]\in \Omega\left(\mathbb{E}_{s,t,t'}\left[ \Pr_{\pi_U}(Y^{(s)}_t \in M, Y^{(s)}_{t+t'} \notin M) \right]\right),
\end{equation}
where $(X^{(s)})$ corresponds to the continuous-time Markov chain with transition kernel $P(s)^2-I$.
We arrive at Eq.~\eqref{eqsup:required-lower-bound-expected-prob} with the help of two lemmas. For the first lemma we prove the following:
\begin{lemma}
\label{lemsup:ctrw-dtrw}
\textit{Consider a reversible continous-time Markov chain $(X)$ with transition matrix $P - I$ and stationary distribution $\pi$.
Let $(Y)$ denote the discrete-time Markov chain with transition matrix $P$.
For any $T$ we have that}
\begin{equation}
\int_0^{40T} \int_0^{40T} \Pr_{\pi_U}(X_t \in M, X_{t+t'} \notin M) \dif t \, \dif t'
\in \Omega\left( \sum_{t,t' = 1}^T \Pr_{\pi_U}(Y_t \in M, Y_{t+t'} \notin M) \right). \nonumber
\end{equation}
\end{lemma}
\begin{proof}
By linearity of expectation it suffices to prove that
\begin{equation} \label{eq:from-v}
\int_0^{40T} \int_0^{40T} \Pr_v(X_t \in M, X_{t+t'} \notin M) \dif t \, \dif t'
\in \Omega\left( \sum_{t,t' = 1}^T \Pr_v(Y_t \in M, Y_{t+t'} \notin M) \right)
\end{equation}
for any $v \in V$.
Denote a single trajectory of the Markov chain $(Y_t)$ by the deterministic sequence $({\cal Y}_t)$, and let $\Pr_v(({\cal Y}_t))$ denote the probability of this trajectory, given that ${\cal Y}_0 = v$.
Then
\begin{equation} \label{eq:traj-DT}
\Pr_v(Y_t \in M, Y_{t+t'} \notin M)
= \sum_{\{({\cal Y}_t)\}} \Pr_v(({\cal Y}_t)) I(\cY_t \in M, \cY_{t+t'} \notin M),
\end{equation}
with $I(Z)$ the indicator function of event $Z$.
We can associate a continuous-time trajectory to $(\cY_t)$ by drawing a sequence of exponentially distributed time intervals $\{\tau_t\}$ of expected length 1, and replacing state $\cY_t$ by a time interval of length $\tau_t$.
For a given sequence $(\cY_t)$ we denote this trajectory by the random variable $(X^{\cY}_t)$.
We can now expand
\begin{equation} \label{eq:traj-CT}
\Pr_v(X_t \in M, X_{t+t'} \notin M)
= \sum_{\{({\cal Y}_t)\}} \Pr_v(({\cal Y}_t)) \Pr(X^\cY_t \in M, X^\cY_{t+t'} \notin M).
\end{equation}
Plugging \eqref{eq:traj-CT} and \eqref{eq:traj-DT} into \eqref{eq:from-v}, it suffices to prove that
\begin{equation}
\int_0^{40T} \int_0^{40T} \Pr_v(X^\cY_t \in M, X^\cY_{t+t'} \notin M) \dif t \dif t'
\in \Omega\left( \sum_{t,t'=1}^T I(\cY_t \in M, \cY_{t+t'} \notin M) \right), \nonumber
\end{equation}
for a single trajectory $({\cal Y}_t)$.

To do so, we will prove a union bound on the following two events:
\begin{itemize}
\item
$Z$ denotes the event that $\sum_{t=1}^{2T} \tau_t \leq 40T$, implying that all states $(\cY_t)_{t \leq T}$ appear as rescaled intervals in the relevant $[0,40T]$ interval of the random variable $(X^\cY_t)$.
\item
$W$ denotes the event that $\left\{ \sum_{t,t'=1}^T \tau_t \tau_{t+t'} \, I(\cY_t \in M, \cY_{t+t'} \notin M) \geq \frac{1}{2} \sum_{t,t'=1}^T I(\cY_t \in M, \cY_{t+t'} \notin M) \right\}$, i.e., the left quantity exceeds $1/2$ times its expectation.
\end{itemize}
Conditioning on these events we get the wanted bound:
\begin{align}
\int_0^{40T} \int_0^{40T} \Pr_v(X^\cY_t \in M, X^\cY_{t+t'} \notin M \mid Z \cap W) \dif t \dif t'
&= \Exp\left[ \sum_{t,t'=1}^T \tau_t \tau_{t+t'} \, I(\cY_t \in M, \cY_{t+t'} \notin M) \,\bigg|\, Z \cap W \right] \nonumber \\
&\geq \frac{1}{2} \sum_{t,t'=1}^T I(\cY_t \in M, \cY_{t+t'} \notin M). \nonumber
\end{align}
It remains to prove that the probability of $Z \cap W$ is large.
First, combining the fact that $\Exp(\sum_{t=1}^{2T} \tau_t) = 2T$ with Markov's inequality we indeed have that $\Pr(Z) \geq 19/20$.
Now define the random variable
\begin{equation}
A = \sum_{t,t'=1}^T \tau_t \tau_{t+t'} \, I(\cY_t \in M, \cY_{t+t'} \notin M). \nonumber
\end{equation}
We wish to bound $\Pr(W) = \Pr(A \geq \Exp[A]/2)$.
We know that $\Exp[A] = \sum_{t,t'=1}^T I(\cY_t \in M, \cY_{t+t'} \notin M)$ and we can bound
\begin{align}
\Exp[A^2]
&= \sum_{t,t',s,s'} \Exp[\tau_t \tau_{t+t'} \tau_s \tau_{s+s'}] \, I(\cY_t \in M, \cY_{t+t'} \notin M)\, I(\cY_s \in M, \cY_{s+s'} \notin M) \nonumber \\
&\leq 4 \left( \sum_{t,t'=1}^T I(\cY_t \in M, \cY_{t+t'} \notin M) \right)^2. \nonumber
\end{align}
To see this, notice that $\Exp[\tau_t \tau_{t+t'} \tau_s \tau_{s+s'}] \, I(\cY_t \in M, \cY_{t+t'} \notin M)\, I(\cY_s \in M, \cY_{s+s'} \notin M) \neq 0$ only if $t \neq t+t'$ and $s \neq s+s'$, and so the ``worst correlation'' shows up when $t = s \neq t+t' = s+s'$, in which case
\begin{equation}
\Exp[\tau_t \tau_{t+t'} \tau_s \tau_{s+s'}]
\leq \Exp[\tau_t^2] \, \Exp[\tau_{t+t'}^2]
= \left(\int_{\tau_t \geq 0} \tau_t^2 e^{-\tau_t} \dif \tau_t\right)^2
= 4, \nonumber
\end{equation}
with the last equality following from partial integration.
Now using the Paley-Zygmund inequality we get that
\begin{equation}
\Pr(W)
= \Pr(A > \Exp[A]/2)
\geq \frac{1}{4} \frac{\Exp[A]^2}{\Exp[A^2]}
\geq \frac{1}{16}. \nonumber
\end{equation}
Finally, by a union bound we get that $\Pr(W \cap Z) \geq 1 - \Pr(\bar W) - \Pr(\bar Z) \geq 1 - 15/16 - 1/20 \geq 1/80$, and so
\begin{align}
\Exp\left[ \sum_{t,t'=1}^T \tau_t \tau_{t+t'} \, I(\cY_t \in M, \cY_{t+t'} \notin M) \right]
&\geq \frac{1}{80} \Exp\left[ \sum_{t,t'=1}^T \tau_t \tau_{t+t'} \, I(\cY_t \in M, \cY_{t+t'} \notin M) \,\bigg|\, Z \cap W \right] \nonumber \\
&\geq \frac{1}{160} \sum_{t,t'=1}^T I(\cY_t \in M, \cY_{t+t'} \notin M). \nonumber
\qedhere
\end{align}
\end{proof}
Note that as Lemma \ref{lemsup:ctrw-dtrw} holds for any reversible Markov chain $P$, it immediate extends to interpolated Markov chains, $P(s)$ for any $s\in [0,1)$. Furthermore, the lemma also holds if we replace $P$ with $P^2$, where $P$ is any reversible Markov chain. Thus, Lemma \ref{lemsup:ctrw-dtrw} already gives us a relation between continuous-time Markov chains $(X^{(s)})$, with walk operator $e^{P^2(s)-I}$  and the Markov chain quantity of interest discrete-time Markov chains $(Z^{(s)})$, with walk operator $P^2(s)$. That is we obtain,
\begin{equation}
\int_0^{40T} \int_0^{40T} \Pr_{\pi_U}(X^{(s)}_t \in M, X^{(s)}_{t+t'} \notin M) \dif t \, \dif t'
\in \Omega\left( \sum_{t,t' = 1}^T \Pr_{\pi_U}(Z^{(s)}_t \in M, Z^{(s)}_{t+t'} \notin M) \right). \nonumber
\end{equation}

Now, all that is left is to relate discrete-time Markov chains with transition matrix $P$ and discrete-time Markov chains with transition matrix $P^2$, for any reversible Markov chain $P$. For this we state the following lemma: 

\begin{lemma}
\label{lemsup:dtrw-dtrw-square}
\textit{Consider a reversible discrete-time Markov chain $(Y)$ with lazy transition matrix $P$ and stationary distribution $\pi$.
Let $(Z)$ denote the discrete-time Markov chain with transition matrix $P^2$.
For any $T$ we have that}
\begin{equation}
\sum_{t,t' = 1}^T \Pr_{\pi_U}(Z_t \in M, Z_{t+t'} \notin M)
\in \Omega\left( \sum_{t,t' = 1}^T \Pr_{\pi_U}(Y_t \in M, Y_{t+t'} \notin M) \right). \nonumber
\end{equation}
\end{lemma}
\begin{proof}
First note that for any $t>0$, $Z_t$ and $Y_{2t}$ are equally distributed, which implies that
\begin{equation}
\sum_{t,t' = 1}^T \Pr_{\pi_U}(Z_t \in M, Z_{t+t'} \notin M)
= \sum_{t,t' = 1}^T \Pr_{\pi_U}(Y_{2t} \in M, Y_{2t+2t'} \notin M)
\eqqcolon \sum_{t,t'=1}^T A_{2t,2t'}. \nonumber
\end{equation}
We would like to show that $\sum_{t,t'=1}^T A_{2t,2t'} \in \Omega\left( \sum_{t,t'=1}^T A_{t,t'} \right)$.
To this end, recall from Lemma \ref{lemsup:success-prob} the expression
\begin{equation}
A_{2t,2t'}
= \braket{\pi_U| P^{2t} \Pi_M P^{2t'} |1_U}. \nonumber
\end{equation}
Since $P$ is a lazy walk, we can rewrite it as $P = I/2 + Q/2$ for some other Markov chain $Q$.
This allows us to bound
\begin{equation}
A_{2t,2t'}
= \braket{\pi_U| (I/2 + Q/2) P^{2t-1} \Pi_M P^{2t'} |1_U}
\geq \frac{1}{2} \braket{\pi_U| P^{2t-1} \Pi_M P^{2t'} |1_U}
= \frac{1}{2} A_{2t-1,2t'}. \nonumber
\end{equation}
By the same argument we have $A_{2t,2t'} \geq \frac{1}{2} A_{2t,2t'-1}$ and hence $A_{2t,2t'} \geq \frac{1}{4} A_{2t-1,2t'-1}$.
Combining these bounds we see that
\begin{equation}
A_{2t,2t'}
= \frac{A_{2t,2t'} + A_{2t,2t'} + A_{2t,2t'} + A_{2t,2t'}}{4}
\geq \frac{A_{2t,2t'} + A_{2t-1,2t'} + A_{2t,2t'-1} + A_{2t-1,2t'-1}}{16}. \nonumber
\end{equation}
The proof now follows by noting that
\begin{equation}
\sum_{t,t'=1}^T A_{2t,2t'}
\geq \frac{1}{16} \sum_{t,t'=1}^T A_{2t,2t'} + A_{2t-1,2t'} + A_{2t,2t'-1} + A_{2t-1,2t'-1}
\geq \frac{1}{16} \sum_{t,t'=1}^T A_{t,t'}. \nonumber \qedhere
\end{equation}
\end{proof}
Thus, Lemma \ref{lemsup:ctrw-dtrw} and Lemma \ref{lemsup:dtrw-dtrw-square}, combined lead to the required bound in Eq.~\eqref{eqsup:required-lower-bound-expected-prob}. Finally, equipped with this result, we can use Lemma \ref{lem:amb} to obtain:

\begin{corollary} \label{cor:amb-ct}
Let $P$ be a lazy, reversible ergodic discrete-time Markov chain with stationary distribution $\pi$.
Let $(X^{(s)})$ denote the continuous-time Markov chain with transition matrix $P^2(s) - I$.
If $\pi(M) \leq 1/9$ and $T \geq 3 \mathrm{HT}(M)$ then choosing interpolation parameter $s \in \{1 - 1/r \,:\, r = 1,2,\dots,2^{\lceil \log(12T)\rceil}\}$ and times $t,t' \in [960T]$ uniformly at random gives that
\begin{equation}
\mathbb{E}_{s,t,t'}\left[ \Pr_{\pi_U}(X^{(s)}_t \in M, X^{(s)}_{t+t'} \notin M) \right]
\in \Omega(1/\log(T)). \nonumber
\end{equation}
\end{corollary}
Thus, Corollary \ref{cor:amb-ct} combined with Eq.~\eqref{eqsup:lower-bound-succ-prob}, results in
\begin{equation}
\mathbb{E}\left[\nrm{\Pi_Me^{(D^2(s)-I)T}\ket{\pi_U}}^2\right]\in \Omega\left(1/(\log^2 T\right), \nonumber
\end{equation}
for $T\in\Theta\left(HT(P,M)\right)$, which proves Lemma 2.
\section{Analog procedure to prepare the ground states of a Hamiltonian}
In this section, we will show that the analog procedure leading to Eq.~\eqref{eqmain:evolved-state-fast-forwarding-3} prepares a quantum state that is $\eps$-close to the ground state of any Hamiltonian.\\~\\
\textbf{Assumptions --} Suppose we have a Hamiltonian $H$ with ground state $\ket{v_0}$ and ground energy $\lambda_0$, and assume that we are given a lower bound on the gap between the ground state and the first excited state of $H$, i.e.\ we are given $\Delta$ such that $|\lambda_1-\lambda_0|\geq \Delta$.\\ 
For clarity of exposition, we assume that the ground space of $H$ is non-degenerate. If this is not the case, e.g.~if the degeneracy of the ground space is $d$ and is spanned by mutually orthonormal eigenstates $\{\ket{v^{(\ell)}_0}\}_{l=1}^{d}$, then we will be preparing a quantum state $\ket{v_0}$ which is a projection onto the ground space given by
$$
\ket{v_0}=\dfrac{1}{\sqrt{\sum_{\ell=1}^{d} |c^{(\ell)}_0|^2}}\sum_{\ell=1}^d c^{(\ell)}_0\ket{v^{(\ell)}_0}.
$$ 
In addition, suppose we have access to some initial state $\ket{\psi_0}$ and a lower bound on the overlap $|\braket{\psi_0|v_0}|=c_0\geq \eta$. \\
Furthermore, for some desired accuracy $\eps\in (0,1)$, we will assume that we know the value of the ground energy to some precision parameter $\eps_g$ such that $\eps_g\in \Oo\left(\Delta/\sqrt{\log\frac{1}{\eta\eps}}\right)$. That is, we know some $E_0$ such that 
\begin{equation}
|\lambda_0-E_0| \leq \eps_g. \nonumber
\end{equation}
By implementing $H-(E_0-\eps_g)I$, we ensure that $0\leq\lambda_0\leq 2\eps_g$. This transformation also ensures that the lower bound for the spectral gap of $H$ remains $\Delta$. If also an upper bound on the maximum eigenvalue of $H$ is known, then we can actually assume that the spectrum of $H$ is in $[0,1]$.\\
Note that these assumptions are standard and apply to quantum algorithms for this problem designed in the circuit model \citep{ge2019faster,lin2020near}.\\~\\
\textbf{Running time and success probability --} Now we are in a position to formally state the ground state preparation algorithm and analyze its complexity, which we do in the following lemma.
\begin{lemma}
\label{lemsup:gsp-ham}
\textit{Suppose $\eps\in (0,1)$ and $\eta\leq 1/\sqrt{2}$. Furthermore suppose we have a Hamiltonian $H$ with ground state $\ket{v_0}$ with $\Delta$ being a lower bound on the spectral gap. Also, the ground state energy of $H$ is known up to a precision $\eps_g\in \Oo\left(\Delta/\sqrt{\log\frac{1}{\eta\eps}}\right)$. Then, given an initial state $\ket{\psi_0}$ satisfying $|\braket{\psi_0|v_0}|\geq \eta$, we output, with probability $\Oo(\eta^2)$, a state $\ket{\phi}$ such that $\nrm{\ket{\phi}-\ket{v_0}}\leq \eps$ by evolving the Hamiltonian $H'=H\otimes \hat{z}$ for time}
$$
T\in \Oo\left(\dfrac{1}{\Delta}\sqrt{\log \left(\dfrac{1}{\eta\eps}\right)}\right).
$$
\end{lemma}
\begin{proof}
We shift the overall eigenvalues of $H$ by $-E_0+\eps_g$ as explained previously. So now suppose $H$ has a spectral decomposition is $H=\sum_{j}\lambda_j\ket{v_j}\bra{v_j}$ with eigenvalues $\lambda_j\in (0,1]$ and particularly $0\leq \lambda_0\leq 2\eps_g$. 

Then the input quantum state, when expressed in the eigenbasis of $H$ is written as $\ket{\psi_0}=\sum_{j}c_j\ket{v_j}$. Without loss of generality, we assume that $c_0$ is real and positive. Then following the analog procedure outlined in the main article, we obtain from Eq.~\eqref{eqmain:evolved-state-fast-forwarding-3} that
\begin{equation}
\ket{\eta_t}=e^{-tH^2}\ket{\psi_0}\ket{\psi_g}+\ket{\Phi}^\perp. \nonumber
\end{equation} 
By post-selecting on obtaining the Gaussian state $\ket{\psi_g}$ on the second register, we are left with the following quantum state in the first register after a time $\sqrt{2t}$
\begin{equation}
\ket{\phi}=\dfrac{e^{-tH^2}\ket{\psi_0}}{\sqrt{\braket{\psi_0|e^{-2tH^2}|\psi_0}}}, \nonumber
\end{equation} 
with probability $\| e^{-tH^2}\ket{\psi_0}\ket{\psi_g} \|^2 = \braket{\psi_0|e^{-2tH^2}|\psi_0}$. Expressing $\ket{\phi}$ in the eigenbasis of $H$, we obtain
\begin{equation}
\ket{\phi}=\dfrac{c_0e^{-t\lambda^2_0}}{\sqrt{\braket{\psi_0|e^{-2tH^2}|\psi_0}}}\left[\ket{v_0}+\sum_{j\geq 1}\dfrac{c_j}{c_0} e^{-t(\lambda^2_j-\lambda^2_0)}\ket{v_j}\right], \nonumber
\end{equation}
where the normalization factor
\begin{equation}
\braket{\psi_0|e^{-2tH^2}|\psi_0}=|c_0|^2 e^{-2t\lambda^2_0}\left[1+\sum_{j\geq 1}\frac{|c_j|^2}{|c_0|^2}e^{-2t(\lambda^2_j-\lambda^2_0)}\right]. \nonumber
\end{equation}
Now, we intend to choose a value of $t$, so that $\ket{\phi}$ and the ground state $\ket{v_0}$ are $\eps$-close to each other in $\ell_2$-norm. We have,
\begin{equation}
\nrm{\ket{\phi}-\ket{v_0}}^2 = 2-2\braket{\phi|v_0}, \nonumber
\end{equation}
where we use the fact that $c_0\geq 0$ which implies $\braket{\phi|v_0}> 0$. Thus we have,
\begin{align}
\braket{\phi|v_0}&=\left[1+\sum_{j\geq 1} \dfrac{|c_j|^2}{|c_0|^2}e^{-2t(\lambda_j^2-\lambda_0^2)}\right]^{-1/2} \nonumber\\
                   &\geq \left[1+\dfrac{1-\eta^2}{\eta^2} e^{-2t(\lambda_1^2-\lambda_0^2)}\right]^{-1/2} \nonumber\\
                   &\geq 1-\dfrac{1-\eta^2}{2\eta^2} e^{-2t(\lambda_1^2-\lambda_0^2)}, \nonumber
\end{align}
which gives   
\begin{equation}
\nrm{\ket{\phi}-\ket{v_0}}^2\leq \dfrac{1-\eta^2}{\eta^2} e^{-2t(\lambda_1^2-\lambda_0^2)}\leq \dfrac{1-\eta^2}{\eta^2} e^{-2t(\lambda_1-\lambda_0)^2}. \nonumber
\end{equation}
Thus by choosing any value of $t$ such that
\begin{align}
t>\dfrac{1}{2\Delta^2}\log\left(\dfrac{1-\eta^2}{\eta^2\eps^2}\right), \nonumber
\end{align}
we ensure that $\nrm{\ket{\phi}-\ket{v_0}}\leq \eps$. The total evolution time of the interaction Hamiltonian is $T=\sqrt{2t}$, which is
\begin{equation}
T\in \Oo\left(\dfrac{1}{\Delta}\sqrt{\log \left(\dfrac{1}{\eta\eps}\right)}\right). \nonumber
\end{equation}
Now we have,
\begin{equation}
\braket{\psi_0|e^{-2tH^2}|\psi_0}= \dfrac{|c_0|^2 e^{-2t\lambda^2_0}}{|\braket{\phi|\psi_0}|^2}\geq |c_0|^2 e^{-2t\lambda^2_0}. \nonumber
\end{equation}
So the success probability of our algorithm
\begin{equation}
\braket{\psi_0|e^{-T^2H^2}|\psi_0}\geq |c_0|^2 e^{-T^2\lambda^2_0} \geq \Oo(\eta^2), \nonumber
\end{equation}
where we have used the fact that $\lambda_0\leq 2\eps_g$ and so $T\lambda_0\in \Oo(1)$.
\end{proof}
The running time of this simple analog procedure matches that of the best-known quantum algorithms for this problem in the circuit model \cite{ge2019faster,lin2020near}. 
\end{document}